\documentclass[preprint]{article}

\usepackage{neurips_2026}
\usepackage[utf8]{inputenc}
\usepackage[T1]{fontenc}
\usepackage{hyperref}
\usepackage{url}
\usepackage{booktabs}
\usepackage{amsfonts}
\usepackage{nicefrac}
\usepackage{microtype}
\usepackage{xcolor}
\usepackage{amsmath}
\usepackage{amssymb}
\usepackage{graphicx}

\usepackage{tikz}
\usetikzlibrary{arrows,cd,decorations.markings,shapes.geometric,shapes}
\usetikzlibrary{positioning,calc,arrows.meta,fit}
\usetikzlibrary{decorations.pathreplacing}

\usepackage{pgfplots}
\usepackage{subcaption}
\usetikzlibrary{positioning,decorations.pathreplacing}
 \pgfplotsset{compat=1.18}

\usepackage{algorithm}
\usepackage{algpseudocode}

\usepackage{microtype}
\usepackage{graphicx}
\usepackage{subcaption}
\usepackage{booktabs} % for professional tables

\usepackage{amsmath}
\usepackage{amssymb}
\usepackage{mathtools}
\usepackage{amsthm}

\theoremstyle{plain}
\newtheorem{theorem}{Theorem}[section]
\newtheorem{proposition}[theorem]{Proposition}
\newtheorem{lemma}[theorem]{Lemma}

\newtheorem{definition}[theorem]{Definition}

\theoremstyle{remark}
\newtheorem{remark}[theorem]{Remark}

\usepackage{amsmath,amssymb,amsthm,mathtools,booktabs}

\usepackage{tikz}
\usepackage{xcolor}

\usepackage{tcolorbox}

\usepackage{algorithm}              % the float
\usepackage{xcolor}

\usepackage{tikz-cd}
\usetikzlibrary{arrows.meta,positioning,calc,decorations.pathreplacing,shapes.geometric}
\usepackage{pgfkeys}
\usepackage{pgfplots}
\usepackage{amsthm}
\usepackage{stmaryrd}   % for \Vdash
\usepackage{booktabs}   % pretty tables
\usepackage{tabularx}   % flexible-width tables
\usepackage{microtype}  % nicer text

\providecommand{\De}

\pgfplotsset{compat=1.18}

\newcommand{\Kl}{\mathsf{Kl}} % Kleisli category notation
\newcommand{\Dist}{\mathsf{Dist}}  % Distribution (probability) monad
\title{Causal Density Functions}

\author{%
  Sridhar Mahadevan\thanks{Academic affiliation: Research Professor, University of Massachusetts, Amherst; See webpage at \url{https://people.cs.umass.edu/~mahadeva/Site/About_Me.html}} \\
  Adobe Research\\
  San Jose, CA\\
  \texttt{smahadev@adobe.com} \\
}

\begin{document}

\maketitle

\begin{abstract}
  We study the full density ratio between a specified intervention regime
  \(P_a\) and an observational regime \(P_0\),
  \(\rho_a=dP_a/dP_0\), under the prerequisite \(P_a\ll P_0\).
  We call the regime-indexed ratio a \emph{causal density function} when
  \(P_a\) is an identified or directly observed interventional law.  The
  underlying Radon--Nikodym derivative and the identity
  \[
  \mathbb{E}_{a}[f(Z)]
  =
  \mathbb{E}_{0}\!\left[f(Z)\rho_a(Z)\right]
  \]
  are classical importance weighting, not new identification results.
  Our narrower question is whether retaining the entire pointwise ratio is
  useful as a reusable diagnostic across several downstream functionals.
  We evaluate a two-density plug-in baseline through held-out moment
  transport and overlap stress tests on synthetic and perturbation data.
  We also report a pairwise graph-scoring heuristic as a negative result:
  its F1 is \(0.10\) on a synthetic DAG, \(0.12\) on Sachs, and \(0.33\) on
  a multi-regime chain.  These experiments do not establish an estimation
  advantage over direct density-ratio, inverse-probability, Riesz, or
  doubly robust methods; they instead delimit what the pointwise ratio
  target and the present plug-in estimator do and do not provide.
\end{abstract}

\section{Introduction}

Causal inference \citep{rubin-book,pearl-book} asks how a system changes under
interventions.  A common target is an interventional expectation such as
\[
  \mathbb{E}[Y\mid \mathrm{do}(X=x)].
\]
This paper studies a more ambitious object: the full pointwise ratio between
an interventional law and an observational law.  Let \(P_0=P_{\mathrm{obs}}\)
and let \(P_a\) be the law under a specified regime \(a\).  If
\(P_a\ll P_0\), define
\[
  \rho_a(z)
  =
  \frac{dP_a}{dP_0}(z).
\]
This is the ordinary Radon--Nikodym density ratio.  We use the term
\emph{causal density function} only when \(P_a\) is a causally meaningful
law that has already been identified from explicit assumptions or is observed
through randomized/controlled perturbation data.  The derivative does not
itself identify \(P_a\) from \(P_0\).

For every integrable \(f\), the classical change-of-measure identity gives
\[
  \mathbb{E}_{a}[f(Z)]
  =
  \mathbb{E}_{0}\!\left[f(Z)\rho_a(Z)\right].
\]
This is importance weighting and is not novel.  It nevertheless supplies an
auditable diagnostic: on held-out samples, one can test whether the same
estimated ratio transports several prespecified moments, compare those gaps
with the corresponding unweighted gaps, and report uncertainty.  Passing a
finite collection of moment tests does not establish equality of measures;
failing them is direct evidence not to trust downstream reuse of the ratio.

Estimating the full ratio is harder than estimating a single functional and is
not automatically preferable.  Separate estimates of \(p_a\) and \(p_0\)
compound nuisance error and can be unstable where \(p_0\) is small.  Direct or
classifier-based density-ratio estimation, inverse-probability weighting,
outcome regression, Riesz-representer methods, and cross-fitted doubly robust
estimators may be superior for their respective targets.  The normalizing-flow
plug-in studied here is therefore a transparent baseline, not a consequence
of the RN theorem and not a claimed state-of-the-art estimator.

Our contributions are:
\begin{enumerate}
  \item We formulate the full intervention/observation ratio as a
        regime-indexed output while making explicit that its analytic basis is
        classical and its causal meaning is conditional on prior identification.
  \item We give a held-out moment-transport protocol that treats overlap,
        unweighted gaps, and uncertainty as part of the report rather than
        interpreting a fitted ratio by inspection alone.
  \item We evaluate a two-density plug-in baseline across synthetic and real
        perturbation regimes and document substantial failures under weak
        overlap.
  \item We report that a proposed pairwise residual scorer is not competitive
        for graph recovery, thereby separating possible reuse of a calibrated
        ratio from unsupported causal-discovery claims.
\end{enumerate}

%==================== Section: Causal Density Functions ====================
\section{Causal Density Functions}
\label{sec:cdf}

The construction requires two probability laws on the same measurable outcome
space.  It does not supply either law or the assumptions needed to identify an
interventional law from observational data.  We use standard
change-of-measure notation \citep{billingsley1995probability,halmos1950measure,rudin1987real}.

\begin{definition}[Causal density function]
Let \(Z=(X_1,\ldots,X_d)\) take values in a measurable space
\((\mathcal Z,\mathcal A)\).  Let \(P_0\) be an observational law and let
\(\{P_a:a\in\mathcal I\}\) be a specified family of intervention or
perturbation laws on that same space.  If \(P_a\ll P_0\), the
\emph{causal density function} for regime \(a\) is
\[
  \rho_a(z)
  \;=\;
  \frac{dP_a}{dP_0}(z).
\]
\end{definition}

\paragraph{Identification and domination.}
The adjective ``causal'' is justified only after \(P_a\) has been identified.
With randomized perturbation samples, \(P_a\) is the empirical regime law.  With
observational data alone, assumptions such as consistency, positivity and
conditional exchangeability, or a valid do-calculus identification formula,
must come first.  A deterministic hard intervention on a continuously
distributed coordinate is commonly singular relative to the observational
joint law.  The joint ratio above is then undefined.  Our experiments therefore
concern observed perturbation regimes, dominated soft/stochastic interventions,
or selected outcome marginals on which domination is plausible.

Absolute continuity is a population prerequisite; empirical overlap is a
separate estimation issue.  Even when \(P_a\ll P_0\), large or heavy-tailed
weights make finite-sample estimation unstable.

\paragraph{Interpretation.}
\(\rho_a(z)>1\) marks regions receiving more probability mass under regime
\(a\), while \(\rho_a(z)<1\) marks depleted regions.  Retaining this pointwise
ratio can support several downstream integrals, but doing so estimates a more
difficult object than any one scalar estimand.

The defining property is the Radon--Nikodym calibration identity.  For any
integrable statistic \(f\),
\begin{equation}
  \mathbb{E}_{a}[f(Z)]
  =
  \mathbb{E}_{0}\!\left[f(Z)\rho_a(Z)\right].
  \label{eq:cdf-calibration}
\end{equation}
Equation~\eqref{eq:cdf-calibration} characterizes the exact ratio over a
measure-determining class of functions.  For an estimated ratio, a finite
collection of held-out moment gaps is only a diagnostic for those functions,
not proof of distributional equality.

\paragraph{Three uses.}
Causal density functions can be reused in three computations considered here:
\begin{enumerate}
  \item \emph{Calibration}: test the change-of-measure identity in
        Equation~\eqref{eq:cdf-calibration}.
  \item \emph{Regime-response curves}: estimate expectations along a dominated
        family \(a_x\) of soft or stochastic interventions.
  \item \emph{Exploratory directed scores}: rank ordered variable pairs using
        a ratio-weighted residual heuristic.  This is not a consistent
        structure-learning procedure in the general multivariate case.
\end{enumerate}

\paragraph{Causal density field.}
Collecting the intervention-specific densities yields a regime-indexed field
\[
  \rho(z) = (\rho_a(z):a\in\mathcal I).
\]
This notation is organizational: each component remains an ordinary density
ratio with its own domination and estimation conditions.

\paragraph{Example (finite case).}
For environments $e_0$ (observational) and $e_1$ (interventional) with empirical
distributions $p_0(x)$ and $p_1(x)$ on a finite set $X$,
\[
  \rho(x) = \frac{p_1(x)}{p_0(x)},
\qquad
  \mathbb{E}_{e_1}[f]
  =
  \sum_x f(x)\rho(x)p_0(x).
\]
This is the finite version of the entire proposal: an intervention is represented
by a density ratio, and causal quantities are computed by reweighted
observational expectations.

\section{Analytic Distinctions and Statistical Scope}
\label{sec:analytic-scope}

\paragraph{Intervention ratio versus dependence ratio.}
The intervention ratio \(dP_a/dP_0\) must not be confused with the
observational dependence density
\[
  \frac{dP_{XY}}{d(P_X\otimes P_Y)}(x,y).
\]
The latter equals one under independence and enters mutual information, but it
is neither directed nor causal without additional structure.  We do not use it
as an alternative definition of a causal density.

\paragraph{Finite regime changes versus infinitesimal scores.}
Suppose
\[
  P_\lambda(dx,dy)=Q_\lambda(dx)K_\lambda(x,dy),\qquad
  P_0(dx,dy)=Q_0(dx)K_0(x,dy).
\]
Under the relevant domination conditions,
\[
  \frac{dP_\lambda}{dP_0}(x,y)
  =
  \frac{dQ_\lambda}{dQ_0}(x)
  \frac{dK_\lambda(x,\cdot)}{dK_0(x,\cdot)}(y).
\]
If a soft intervention changes only the input distribution and preserves the
conditional mechanism \(Y\mid X\), the second factor is one.  Separately, for
a differentiable density path, the score is
\(s_\lambda(z)=\partial_\lambda\log p_\lambda(z)\), and
\[
 \log\frac{p_{\lambda+\delta}(z)}{p_\lambda(z)}
 =\delta s_\lambda(z)+o(\delta).
\]
This local expansion is not an identity between a finite regime ratio and the
derivative of an observational dependence ratio.

\paragraph{What ratio consistency requires.}
Pointwise convergence of fitted log-densities does not imply
\(L^1(P_0)\) convergence of their ratio: a fitted denominator may approach zero
on shrinking sets and create arbitrarily large weights.  A valid sufficient
route is convergence \(\widehat\rho_n\to\rho\) in \(P_0\)-probability together
with uniform integrability of \(\{\widehat\rho_n\}\); Vitali's theorem then gives
\[
  \|\widehat\rho_n-\rho\|_{L^1(P_0)}\longrightarrow 0.
\]
Stronger bounded-ratio or uniform log-density assumptions can imply these
conditions, but ordinary pointwise consistency cannot.

\paragraph{What is not claimed.}
The paper makes no categorical universal-property claim.  Kan extensions,
Beck--Chevalley conditions, and a proposed ``RN--Kan duality'' do not provide
the identification, estimator, or consistency guarantee used here and are
therefore omitted.  We also do not claim that the pairwise graph score controls
for confounding, mediators, or multiple parents.

\section{Estimating Causal Density Functions}
\label{sec:kan-do-algms}

We now describe the estimator used in the experiments.  The goal is deliberately
narrow: estimate causal density ratios, test their calibration, and use them to
construct regime-response curves or exploratory directed scores.

\paragraph{Density-ratio estimation.}
For an intervention family \(a=\mathrm{do}(X_i)\), we fit or otherwise estimate
two densities on a shared support: an observational density
\(\hat p_{\mathrm{obs}}\) and an interventional density \(\hat p_a\).  The
plug-in causal density estimate is
\[
  \hat\rho_a(z)
  =
  \exp\!\left(
    \log \hat p_a(z) - \log \hat p_{\mathrm{obs}}(z)
  \right).
\]
In the experiments below, low-dimensional normalizing flows estimate the two
densities.  This plug-in is only one baseline and may be less stable than direct
or classifier-based ratio estimation because error in both densities enters
the quotient.  Inverse-probability and doubly robust methods can be preferable
when their identification assumptions and target functionals apply.

\paragraph{Calibration.}
We fit \(\hat\rho_a\) without using the evaluation fold.  For a prespecified
statistic \(f\), held-out samples give the weighted moment gap
\[
  \widehat\Delta^{\,w}_f(a)
  =
  \left|
  \widehat{\mathbb{E}}_{0}\!\left[f(Z)\hat\rho_a(Z)\right]
  -
  \widehat{\mathbb{E}}_{a}[f(Z)]
  \right|.
\]
We compare it with the unweighted gap
\[
  \widehat\Delta^{\,0}_f(a)
  =
  \left|\widehat{\mathbb{E}}_{0}[f(Z)]
  -\widehat{\mathbb{E}}_{a}[f(Z)]\right|.
\]
Bootstrap intervals or repeated splits should quantify uncertainty.  A finite
collection of small gaps validates only those transported moments; it does not
certify the entire ratio.

\paragraph{Do-curves.}
For a pair \((X_i,Y)\), let \(\{a_x:x\in\mathcal X\}\) be a dominated family of
identified or directly sampled stochastic interventions.  The response curve is
\[
  \widehat m_{i\to Y}(x)
  =
  \mathbb{E}_{\mathrm{obs}}\!\left[
    Y\,\hat\rho_{a_x}(Z)
  \right].
\]
This formula neither identifies \(P_{a_x}\) nor justifies deterministic
\(\mathrm{do}(X_i=x)\) when the resulting law is singular.

\paragraph{Directed edge scoring.}
For causal discovery, we use \(\hat\rho_i\) to score whether an intervention on
\(X_i\) explains \(X_j\).  Let
\(\widehat{\mathbb{E}}[X_j\mid X_i]\) be a fitted conditional mean.  The
pairwise causal-density score is
\[
  s_{ij}
  =
  \mathbb{E}_{\mathrm{obs}}\!\left[
    \hat\rho_i(Z)\,
    \frac{(X_j-\widehat{\mathbb{E}}[X_j\mid X_i])^2}{\mathrm{Var}(X_j)}
  \right].
\]
Lower scores mean only a better pairwise fit under this weighted residual
criterion.  Because the regression omits other variables, the score can confuse
direct effects, mediated associations and common causes.  Sparse selection and
cycle pruning do not repair that limitation.

\begin{algorithm}[t]
\caption{Causal Density Estimation and Scoring}
\label{alg:cdf}
\begin{algorithmic}[1]
\Require Observational samples, interventional/regime samples, variables \(X_1,\dots,X_d\)
\For{each intervention family \(a\)}
  \State Fit \(\hat p_{\mathrm{obs}}\) and \(\hat p_a\) on a shared support.
  \State Compute \(\hat\rho_a(z)=\exp(\log\hat p_a(z)-\log\hat p_{\mathrm{obs}}(z))\).
  \State Evaluate weighted and unweighted held-out gaps for selected \(f\)'s.
\EndFor
\For{each ordered pair \(X_i\to X_j\)}
  \State Fit \(\widehat{\mathbb{E}}[X_j\mid X_i]\).
  \State Compute \(s_{ij}=\mathbb{E}_{\mathrm{obs}}[\hat\rho_i(Z)(X_j-\widehat{\mathbb{E}}[X_j\mid X_i])^2/\mathrm{Var}(X_j)]\).
\EndFor
\State Optionally compute pairwise scores as an exploratory diagnostic.
\State \textbf{Return:} ratio estimates, diagnostics, response curves, and optional scores.
\end{algorithmic}
\end{algorithm}

\begin{proposition}[A sufficient transport condition]
\label{prop:ratio-transport}
Suppose \(\widehat\rho_{a,n}\to\rho_a\) in \(P_0\)-probability and
\(\{\widehat\rho_{a,n}\}\) is uniformly integrable under \(P_0\).  Then
\(\|\widehat\rho_{a,n}-\rho_a\|_{L^1(P_0)}\to0\).  Consequently, for every
bounded measurable \(f\),
\[
 \left|\mathbb E_0[f(Z)\widehat\rho_{a,n}(Z)]
       -\mathbb E_a[f(Z)]\right|\to0.
\]
\end{proposition}
\begin{proof}
Vitali's convergence theorem gives \(L^1(P_0)\) convergence.  The second claim
follows by H\"older's inequality and the RN identity.
\end{proof}

This proposition assumes ratio convergence; it is not a theorem that
normalizing-flow maximum likelihood automatically supplies it.  A result for a
particular learner requires denominator control, nuisance rates, and suitable
sample-splitting or empirical-process conditions.

\section{Experimental Results}

We audit the submitted ratio pipeline rather than present it as a competitive
causal estimator.  We ask how well fitted ratios transport selected held-out
moments, what response curves the plug-in produces, and how the optional
pairwise graph heuristic behaves on known structures.  The archived runs did
not record unweighted gaps, bootstrap intervals, or competitive dose-response
baselines.  We report that missing evidence explicitly.

\subsection{PISA 2022 Socio--Economic Panel}
\label{app:pisa2022}

We use the OECD PISA socio-economic status (ESCS) Trend extract as a
cross-regime stress test.\footnote{The PISA datasets are available at
\url{https://webfs.oecd.org/pisa2022/index.html.}}
We evaluate causal-density estimation on a four--variable socio--economic panel constructed
from PISA~2022.  
Each row corresponds to a country, and the variables are standardized OECD
indices for economic, social, and cultural status, parental education, home
possessions, and occupational status.  Countries are not interventions and the
variables include constructed socioeconomic indices, so this is not a
ground-truth causal benchmark.  Weighted moment transport is poor for the
selected \texttt{hisei\_trend}/\texttt{escs\_trend} pair
(\(\Delta_y\approx4.0\), \(\Delta_{y^2}\approx20.0\)).  This is the intended
failure signal: cross-country regimes are too heterogeneous for reliable
density-ratio transport.  We infer no causal direction from these scores.

\subsection{Sachs Protein Signaling}
\label{sec:sachs11}

We evaluate causal-density edge scoring on the classical \emph{Sachs} flow–cytometry dataset
of protein–signaling pathways in human immune cells 
(11 phospho–proteins/phospho–lipids, single–cell measurements).
Each stimulation or inhibitor setting defines an experimental regime 
(e.g., \texttt{CD3CD28}, \texttt{PKA\_inh}, \texttt{PKC\_act}, \ldots); 
we encode these as an environment label \texttt{env}.
Following the standard public benchmark \citep{sachs}, each condition records a
stimulation or inhibition of signaling nodes:
\begin{center}
\small
\texttt{\{raf, mek, erk, pka, pkc, pip2, pip3, plcg, akt, p38, jnk\}}.
\end{center}

\paragraph{Results.}
Figure~\ref{fig:s9} visualizes the resulting causal-density edge-score matrix
and the per-variable regime-divergence statistics.
Against the 12-edge reference graph in the released evaluation, the pairwise
scorer has SHD \(30\), F1 \(0.1176\), precision \(0.0909\), and recall
\(0.1667\) (TP \(2\), FP \(20\), FN \(10\)).  The regime-penalized variant has
SHD \(16\) and F1 \(0.1111\) (TP \(1\), FP \(5\), FN \(11\)).  These values do
not support a structure-learning claim.

\begin{figure}[t]
  \centering
  \includegraphics[width=.5\linewidth]{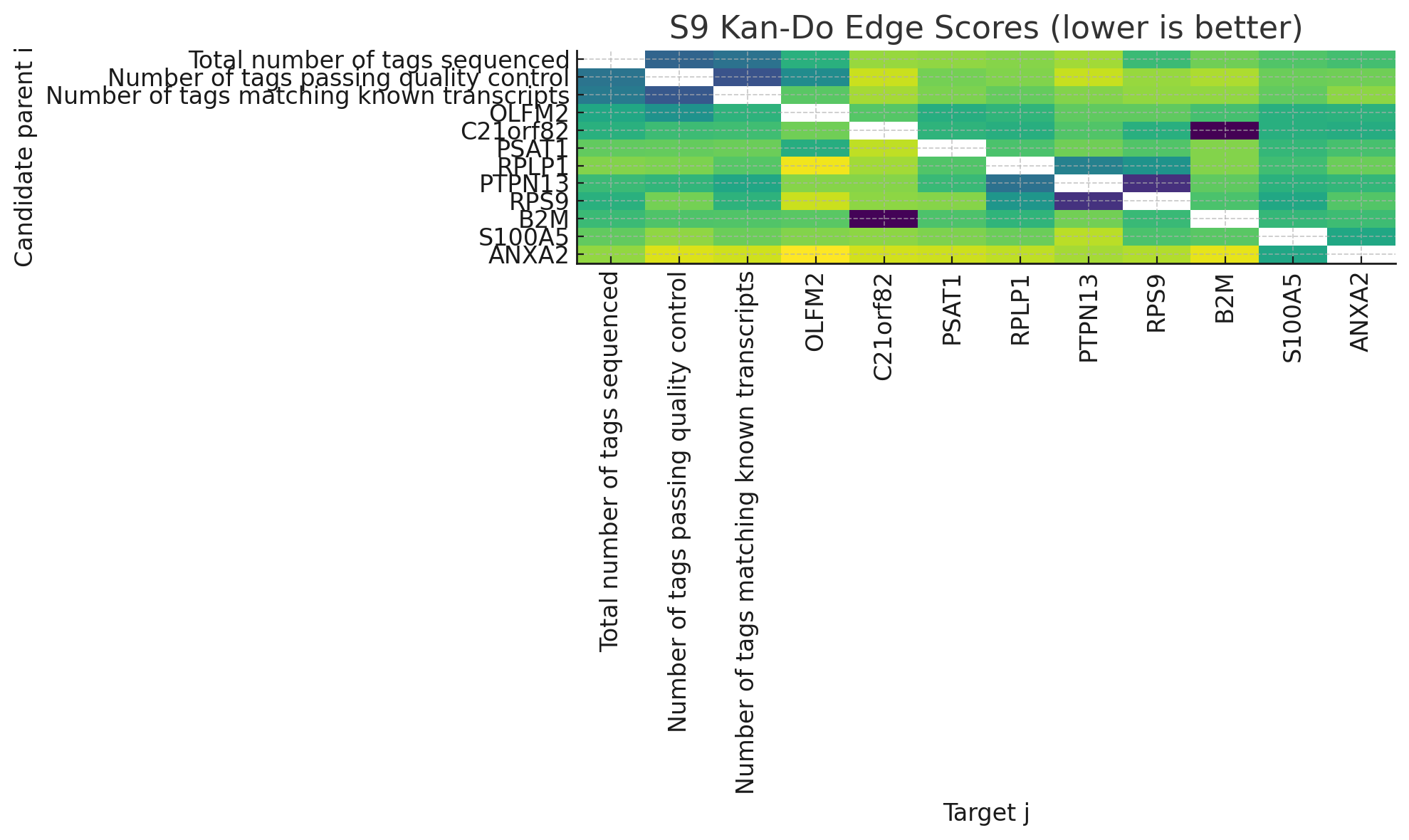}\hfill
  \includegraphics[width=.5\linewidth]{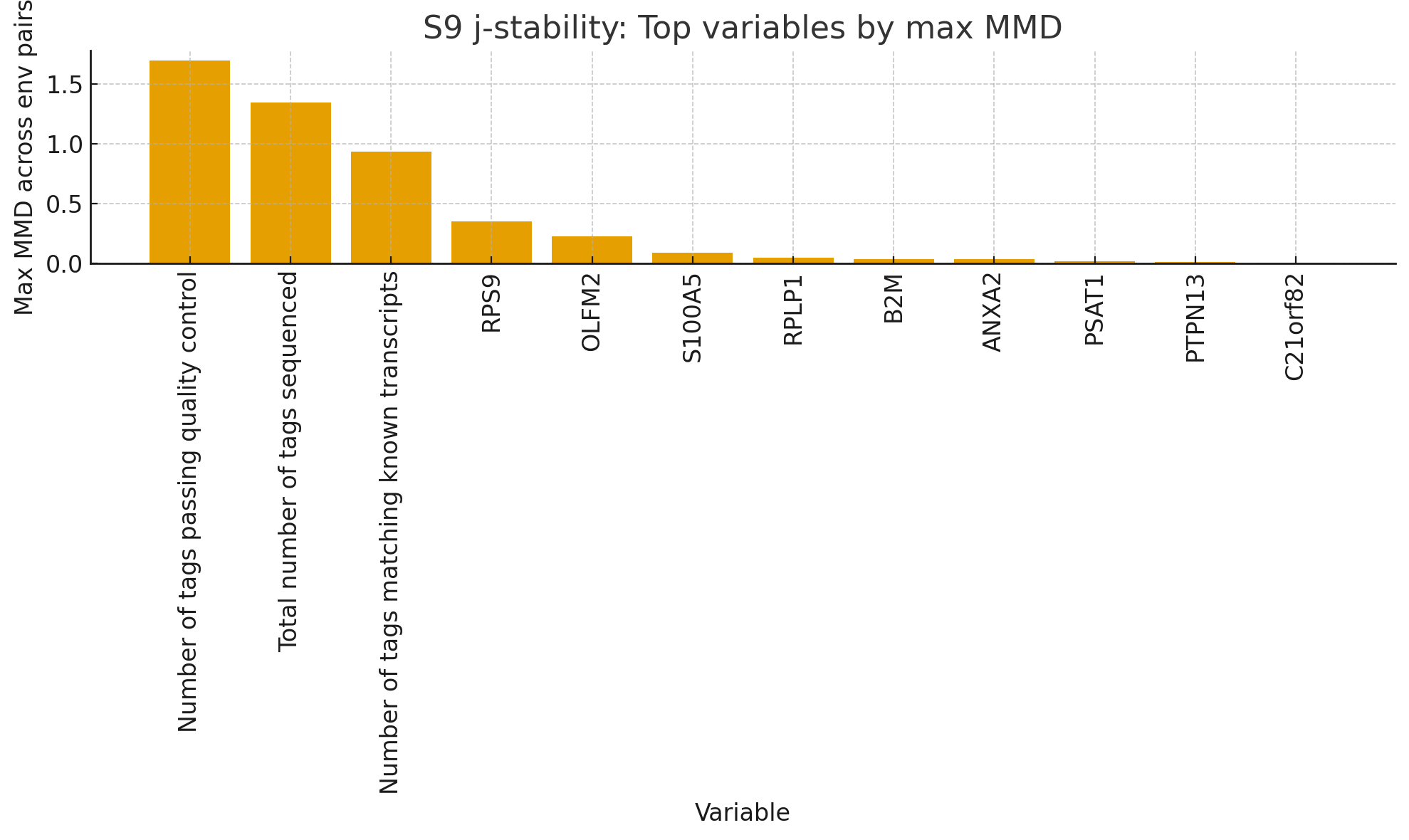}
  \caption{\textbf{Sachs causal-density results.} 
  Left: causal-density edge-score heatmap ($s_{ij}$; lower is better). 
  Right: regime divergence (max MMD per variable across environments). 
  The heatmap is exploratory; the corresponding pairwise graph has F1
  \(0.1176\).}
  \label{fig:s9}
\end{figure}

\subsection{Multi-Regime Chain Experiment}
\label{sec:sheaf}
To test causal-density scoring when multiple regimes are available, 
we construct a synthetic dynamical system in which each regime represents a local 
perturbation of an underlying causal process. The domain consists of $d=8$ scalar variables forming a 
linear chain $X_1 \to X_2 \to \cdots \to X_8$. Within each regime $r \in \{1,\ldots,4\}$, the mechanism 
for each variable is given by 
\[
  X_j^{(r)} = 0.8\,X_{j-1}^{(r)} + \epsilon_j^{(r)}, 
  \qquad \epsilon_j^{(r)} \sim \mathcal{N}(0, 0.4^2),
\]
with soft interventions applied to a small subset of nodes by shifting their means 
(e.g., $X_1 \!+\! 0.8$ in regime 2 and $X_2 \!-\! 0.6$ in regime 3). 
This generates a set of overlapping local models with a known chain ground truth.

Figure~\ref{fig:sheaf_fast} illustrates the resulting adjacency matrices and edge--score field.  
The score matrix contains some low-scoring true pairs but also misses true
edges and introduces false positives.  With the current pairwise scorer and
sparse selection rule,
the recovered graph achieves SHD \(=8\) and F1 \(=0.33\) (Table~\ref{tab:kan-do-structure-recovery}).
This is a negative structure-learning result, not evidence of consistency.

\begin{figure}[t]
  \centering
  \includegraphics[width=.32\linewidth]{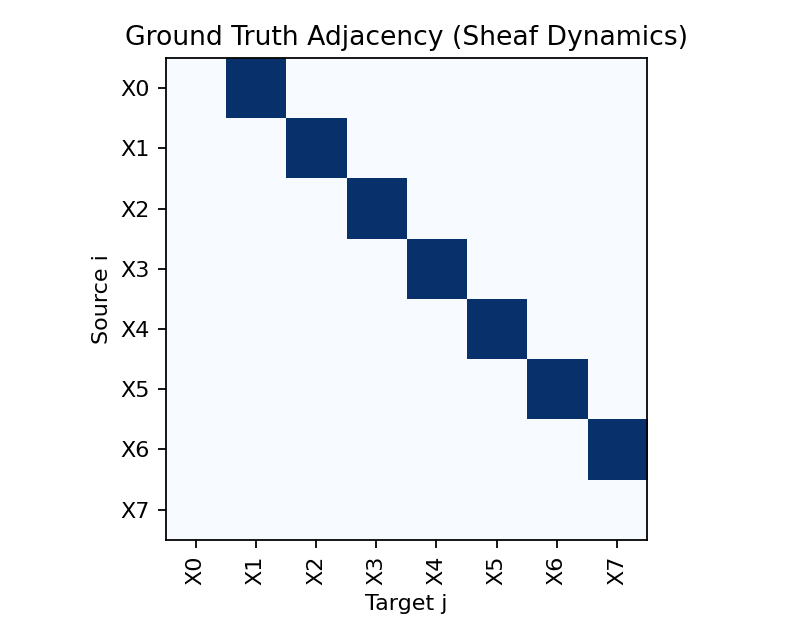}\hfill
  \includegraphics[width=.32\linewidth]{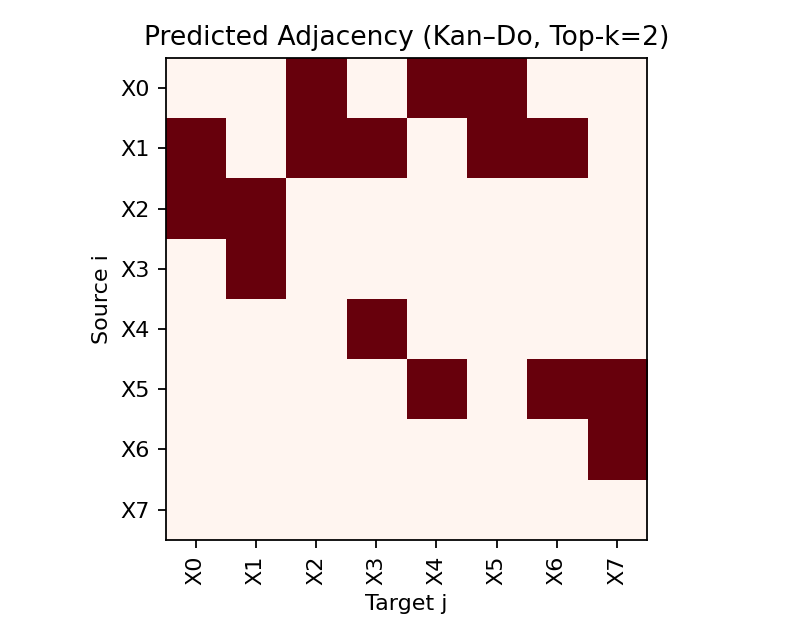}\hfill
  \includegraphics[width=.32\linewidth]{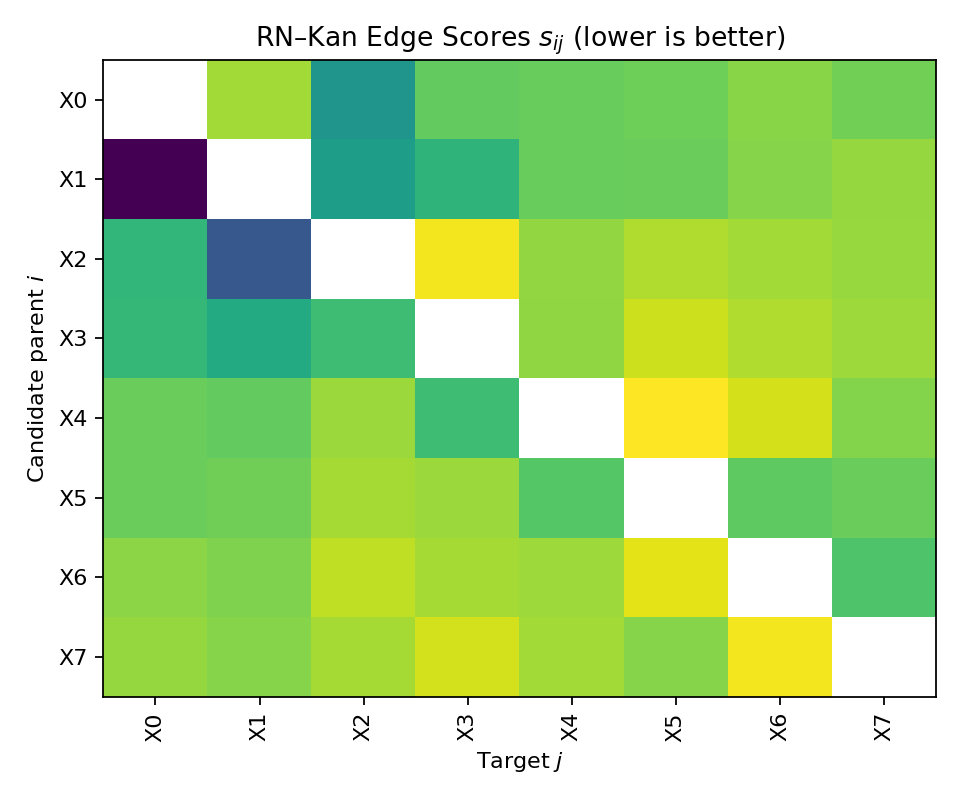}
  \caption{\textbf{Multi-regime chain causal-density scores.} 
  Left: Ground--truth chain structure ($X_0\!\to\!X_1\!\to\!\cdots\!\to\!X_7$). 
  Middle: predicted adjacency from the pairwise causal-density scorer. 
  Right: causal-density edge-score matrix $s_{ij}$ (lower is better). 
  Extra edges remain, but the correct chain directionality is visible in the 
  low-score band near the diagonal.}
  \label{fig:sheaf_fast}
\end{figure}

\begin{table}[t]
\centering
\caption{\textbf{Archived weighted moment-transport gaps.}
Absolute gaps between
$\mathbb{E}_{\mathrm{obs}}[f(Y)\rho]$
and
$\mathbb{E}_{\mathrm{do}}[f(Y)]$
for $f(y)=y$ and $f(y)=y^2$.
Smaller is better, but the archived runs lack unweighted gaps and sampling
intervals and therefore do not establish a benefit from weighting.}
\vspace{0.5em}
\begin{tabular}{lcc}
\toprule
\textbf{Dataset / Pair} & $\bigl|\Delta_{f(y)=y}\bigr|$ & $\bigl|\Delta_{f(y)=y^2}\bigr|$ \\
\midrule
Synthetic Gaussian          & $2.023$ & $0.107$ \\ \hline
LINCS                  & $0.127$      & $0.322$      \\
(HSPA8 $\to$ CDC25B)    & & \\  \hline 
PISA2022      & $4.145$       & $20.005$          \\
(hisei\_trend $\to$ escs\_trend) & & \\
\bottomrule
\end{tabular}
\label{tab:rncalib_comparison}
\end{table}

\paragraph{Cross-domain behavior of RN calibration.}
Table~\ref{tab:rncalib_comparison} reports the values in the released CSVs.
They contradict the earlier manuscript's claim of near-exact synthetic
calibration: the first-moment gap is \(2.023\).  LINCS has gaps \(0.127\) and
\(0.322\); PISA has gaps \(4.145\) and \(20.005\).  No corresponding Sachs
calibration CSV was found, so the previously reported \(0.01/0.02\) values are
withdrawn.  Without unweighted gaps and uncertainty, the table audits plug-in
outputs but does not show that weighting improves transport.

\subsection{Discussion: Causal Density and Regime Overlap}
\label{sec:discussion}

The present results do not show that the pairwise score identifies stable
mechanisms.  Sachs F1 is about \(0.12\), synthetic-DAG F1 is \(0.10\), and
multi-regime-chain F1 is \(0.33\).  The useful lesson is procedural:
downstream reuse should follow held-out, baseline-relative transport checks,
and graph recovery requires multivariate control for alternative parents and
indirect paths.

\subsection{Synthetic Benchmark Experiments}
\label{app:synthetic}

To complement the real--world experiments (Sachs, LINCS, PISA) and the 
multi-regime chain example, we evaluate causal-density scoring on standard 
synthetic causal discovery benchmarks.  
These synthetic tests use linear--Gaussian Erd\H{o}s--R\'enyi DAGs, following
the setting used by NOTEARS, DCDI, CAM, and related structure-learning methods.
Table~\ref{tab:kan-do-structure-recovery} summarizes directed structure
recovery for the current pairwise scorer.

\begin{table}[t]
\centering
\caption{\textbf{Directed structure recovery for causal-density edge scoring on synthetic and real datasets.}}
\begin{tabular}{lcc}
\toprule
Dataset & SHD $\downarrow$ & F1 $\uparrow$ \\
\midrule
Synthetic DAG (ER, d=10)          & 37 & 0.10 \\
Multi-regime chain                & 8  & 0.33 \\
Sachs (pairwise)                  & 30 & 0.12 \\
Sachs (multi + regime penalty)    & 16 & 0.11 \\
\bottomrule
\end{tabular}
\label{tab:kan-do-structure-recovery}
\end{table}

The pairwise scorer is not competitive with mature graph-optimization
methods such as NOTEARS \citep{zheng2018dags} or DCDI
\citep{brouillard2020differentiablecausaldiscoveryinterventional}.  We therefore treat graph recovery as a
failure analysis rather than a causal-discovery contribution.  The current
transport and response-curve results are also incomplete without unweighted
baselines, uncertainty, and competitive ratio or dose-response estimators.

\section{Conclusion}

The defensible proposal is modest: when a dominated interventional law is
already identified or directly sampled, retain its full ratio to an
observational reference as a reusable, regime-indexed object.  Its value must
be demonstrated rather than assumed—through cross-fitted transport tests,
baseline comparisons, uncertainty, and downstream tasks that genuinely benefit
from reuse.  The present plug-in experiments expose important failures and do
not establish such a benefit, but they define a concrete evaluation protocol
for determining when the more ambitious full-ratio target is warranted.

\section{Limitations and Future Work}

We have defined causal density functions as Radon--Nikodym ratios between
specified interventional and observational laws.  This terminology does not
make density ratios novel, identify an interventional law, or avoid positivity
and overlap requirements.  Hard interventions may be singular, in which case
the joint ratio is undefined.  Estimating two full densities can be harder and
less stable than direct ratio, IPW, outcome-regression, Riesz, or doubly robust
alternatives.  The archived transport results lack unweighted baselines and
uncertainty, so they cannot establish practical benefit.  The pairwise graph
score omits multivariate adjustment and performs poorly; it should not be used
as a structure learner.  Finally, the largest experiment has about 30
variables, while reliable full-ratio estimation in high dimension will require
additional structure and much stronger evaluation.

%\bibliography{allcitations}

\begin{thebibliography}{21}
\providecommand{\natexlab}[1]{#1}
\providecommand{\url}[1]{\texttt{#1}}
\expandafter\ifx\csname urlstyle\endcsname\relax
  \providecommand{\doi}[1]{doi: #1}\else
  \providecommand{\doi}{doi: \begingroup \urlstyle{rm}\Url}\fi



\bibitem[Billingsley(1995)]{billingsley1995probability}
Patrick Billingsley.
\newblock \emph{Probability and Measure}.
\newblock Wiley, 3rd edition, 1995.

\bibitem[Brouillard et~al.(2020)Brouillard, Lachapelle, Lacoste, Lacoste-Julien, and Drouin]{brouillard2020differentiablecausaldiscoveryinterventional}
Philippe Brouillard, Sébastien Lachapelle, Alexandre Lacoste, Simon Lacoste-Julien, and Alexandre Drouin.
\newblock Differentiable causal discovery from interventional data.
\newblock In \emph{Advances in Neural Information Processing Systems}, 2020.

\bibitem[Halmos(1950)]{halmos1950measure}
Paul~R. Halmos.
\newblock \emph{Measure Theory}.
\newblock Van Nostrand, 1950.

\bibitem[Imbens and Rubin(2015)]{rubin-book}
Guido~W. Imbens and Donald~B. Rubin.
\newblock \emph{Causal Inference for Statistics, Social, and Biomedical Sciences: An Introduction}.
\newblock Cambridge University Press, USA, 2015.
\newblock ISBN 0521885884.

\bibitem[Pearl(2009)]{pearl-book}
Judea Pearl.
\newblock \emph{Causality: Models, Reasoning and Inference}.
\newblock Cambridge University Press, USA, 2nd edition, 2009.
\newblock ISBN 052189560X.

\bibitem[Rudin(1987)]{rudin1987real}
Walter Rudin.
\newblock \emph{Real and Complex Analysis}.
\newblock McGraw-Hill, 3rd edition, 1987.

\bibitem[Sachs et~al.(2005)Sachs, Perez, Pe'er, Lauffenburger, and Nolan]{sachs}
Karen Sachs, Diego Perez, Dana Pe'er, Douglas~A Lauffenburger, and Garry~P Nolan.
\newblock Causal protein-signaling networks derived from multiparameter single-cell data.
\newblock \emph{Science}, 308\penalty0 (5721):\penalty0 523--529, 2005.
\newblock \doi{10.1126/science.1105809}.

\bibitem[Zheng et~al.(2018)Zheng, Aragam, Ravikumar, and Xing]{zheng2018dags}
Xun Zheng, Bryon Aragam, Pradeep Ravikumar, and Eric~P Xing.
\newblock Dags with no tears: Continuous optimization for structure learning.
\newblock In \emph{Advances in Neural Information Processing Systems}, 2018.

\end{thebibliography}
%\bibliographystyle{plainnat}

%\input{checklist.tex} 

\appendix
\onecolumn
\section*{Supplementary Materials}

\section{Changes in This Version}
\label{app:changes}

This version substantially narrows and corrects the paper.
\begin{enumerate}
\item \textbf{Classical status and causal scope.}
The intervention/observation Radon--Nikodym derivative is now identified
explicitly as a classical density ratio and its expectation identity as
importance weighting.  It does not identify an interventional law.  Causal
interpretation requires an already identified or directly observed regime law,
and the joint ratio requires absolute continuity.  Potentially singular hard
interventions are no longer treated as automatically covered.
\item \textbf{Removal of categorical claims.}
The RN--Kan theorem, Kan--Do calculus, Beck--Chevalley/do-calculus claims, and
categorical interpretation of DCDI have been withdrawn from the manuscript.
They did not derive the estimator and contained typing and naturality errors.
\item \textbf{Corrected analysis.}
The observational dependence ratio is separated from the intervention ratio,
and finite regime changes are separated from infinitesimal path scores.  The
incorrect claim that pointwise log-density convergence implies
\(L^1(P_0)\) ratio convergence is replaced by a sufficient result using
convergence in probability and uniform integrability.
\item \textbf{Repositioned estimation and diagnostics.}
The two-density normalizing-flow plug-in is presented as a baseline rather than
an advantage over direct density-ratio, inverse-probability, Riesz, outcome
regression, or doubly robust methods.  Moment transport is described as a
held-out diagnostic that requires unweighted baselines and uncertainty; finitely
many transported moments do not certify an entire density.
\item \textbf{Corrected empirical record.}
Claims inconsistent with the released artifacts are withdrawn.  In particular,
the archived synthetic moment gaps are \(2.0228\) and \(0.1069\), not
\(0.001\) and \(0.002\); no released file supports the former Sachs
\(0.01/0.02\) calibration claim; and the former Sachs F1 \(0.68\) and chain
F1 \(0.94\)/SHD \(1\) claims are removed.  The verified graph results are
reported as a negative result for the pairwise scorer.
\item \textbf{Bibliographic corrections.}
The invalid ``Ke et al.'' record was removed, the duplicate DCDI entry was
replaced by the correct NeurIPS 2020 citation, and the intervention-permutation
paper's authorship was corrected.
\end{enumerate}

The supplement reports the additional datasets and the archived outputs used in
the audit.  The earlier categorical supplement and its RN--Kan, Kan--Do, and
Beck--Chevalley claims have been removed because they neither derive the
estimator nor survive the required typing checks.

\section{Additional Experimental Details}
\label{app:experiments}

The released CSV files, rather than prose claims in the submitted appendix,
are the source of truth.  Table~\ref{tab:artifact-audit} records the verified
graph metrics.

\begin{table}[h]
\centering
\caption{\textbf{Verified graph-recovery artifacts.}  All metrics are for
directed edges.}
\begin{tabular}{lrrrrrrr}
\toprule
Run & SHD & F1 & Precision & Recall & TP & FP & FN\\
\midrule
Synthetic DAG & 37 & 0.0976 & 0.2857 & 0.0588 & 2 & 5 & 32\\
Sachs pairwise & 30 & 0.1176 & 0.0909 & 0.1667 & 2 & 20 & 10\\
Sachs regime-penalized & 16 & 0.1111 & 0.1667 & 0.0833 & 1 & 5 & 11\\
Chain pairwise & 17 & 0.2609 & 0.1875 & 0.4286 & 3 & 13 & 4\\
Chain regime-penalized & 8 & 0.3333 & 0.4000 & 0.2857 & 2 & 3 & 5\\
\bottomrule
\end{tabular}
\label{tab:artifact-audit}
\end{table}

The previously stated Sachs F1 \(0.68\) and chain F1 \(0.94\)/SHD \(1\)
do not appear in the released metrics and are withdrawn.  Likewise, the
released synthetic calibration CSV contains gaps \(2.0228\) and \(0.1069\),
not \(0.001\) and \(0.002\).  The LINCS and PISA files contain weighted gaps
\((0.1267,0.3222)\) and \((4.1450,20.0047)\), respectively.  The archived
evaluation lacks unweighted moment gaps, resampling intervals, direct-ratio
baselines, and competitive dose-response estimators.  A future experiment
must preregister the statistics \(f\), use cross-fitting, report effective
sample size and weight tails, and compare weighted with unweighted transport.

\end{document}